\documentclass{article}
\usepackage{amsmath,amssymb,graphicx,authblk}

\providecommand{\texorpdfstring}[2]{#1}

\makeatletter
\newcounter{theorem}[section]
\renewcommand{\thetheorem}{\thesection.\arabic{theorem}}
\newenvironment{theorem}[1][]{\refstepcounter{theorem}\par\noindent\textbf{Theorem~\thetheorem.} \textit{#1}\;}{\par}
\newenvironment{lemma}[1][]{\refstepcounter{theorem}\par\noindent\textbf{Lemma~\thetheorem.} \textit{#1}\;}{\par}
\newenvironment{proof}{\par\noindent\textit{Proof. }\normalfont}{\hfill$\square$\par}
\makeatother

\title{Resonant Weighted Nonlocal Schr\"odinger Equation with Gauge Invariance, Conservation Laws and Measurable Phase Detuning}

\author[1]{L. Yıldız\thanks{li.yildiz.na@gmail.com}}
\author[2]{D. Kaykı\thanks{deha.kayki@ogr.iu.edu.tr}}
\author[3]{E. Güdekli\thanks{gudekli@istanbul.edu.tr}}

\affil[1,2,3]{Department of Physics, Faculty of Science, Istanbul University, Istanbul 34134, Turkey}

\begin{document}

\maketitle
\begin{abstract}
We present a gauge-invariant Schr\"odinger-type evolution that combines (i) weighted local diffusion, (ii) symmetric nonlocal exchange through a kernel operator, and (iii) a mean-free phase-resonant drive. The resulting Resonant Weighted Nonlocal Schr\"odinger (RWNS) equation exactly conserves mass and, when the drive is absent, admits a Hamiltonian structure with energy conservation. Under standard assumptions on the weight, kernel, and nonlinearity, we establish local well-posedness in $H^1$ and provide defocusing conditions for global continuation. Linearization yields a dispersion relation in which the nonlocal kernel and the mean-free phase field contribute additively to a measurable spectral detuning. Building on this, we define two observables: a wavenumber-resolved detuning $\Delta\omega(k)$ and a kernel-contrast functional $\Xi[\psi]$ that isolates the nonlocal exchange. We outline feasible implementations in nonlinear-optical lattices and cavity-assisted cold-atom platforms, and discuss conceptual links to propagation-induced phase signatures in astrophysical media. The RWNS model thus offers a compact and analytically tractable framework that unifies weighted local dynamics, symmetric nonlocality, and a mean-free phase drive, yielding clear, testable predictions for laboratory measurements and, in principle, precision timing data.
\end{abstract}

\section{Introduction}
Schr\"odinger-type evolution equations provide a common language for wave propagation across quantum, optical, and cold-atom platforms~\cite{SulemSulem1999,Agrawal2013,MorschOberthaler2006,PitaevskiiStringari2016}. Extensions are often required to represent complex media: heterogeneous (weighted) operators capture spatially varying transport~\cite{Evans2010,GilbargTrudinger2001}; \emph{symmetric} nonlocal exchange models interactions across finite distances via an even kernel~\cite{DuGunzburgerLehoucqZhou2013,BucurValdinoci2016}; and phase-driving terms encode structured pumps or control fields~\cite{Eckardt2017}. In most settings these ingredients are introduced separately, and phase drives may be incorporated in ways that obscure conserved quantities or compromise global gauge invariance~\cite{Cazenave2003,Tao2006}.

Here we present a compact, gauge-invariant framework that couples all three mechanisms within a single evolution law. The model consists of (i) a weighted local diffusion (Laplacian-type) operator, (ii) a symmetric nonlocal exchange generated by an even kernel $K$ (e.g., $K(x)=K(-x)$ in translation-invariant settings, or $K(x,y)=K(y,x)$ in general), and (iii) a \emph{mean-free} phase-resonant drive, i.e., a drive with vanishing spatial mean so that the global $U(1)$ phase symmetry is preserved at all times~\cite{FukushimaOshimaTakeda2011}. We refer to the resulting equation as the \emph{Resonant Weighted Nonlocal Schr\"odinger} (RWNS) model. The formulation is posed on domains $\mathcal{M}\subset\mathbb{R}^d$ (and, where stated, more general measure spaces), accommodates standard nonlinear responses, and is arranged so that each parameter corresponds to an experimentally meaningful control knob.

The RWNS equation has two immediate structural consequences. First, the $L^2$ mass is conserved for all admissible solutions under the stated regularity and symmetry assumptions~\cite{Cazenave2003,Tao2006}. Second, when the phase drive is switched off, the dynamics admit a Hamiltonian structure and conserve energy under standard hypotheses on the weight $w$ and kernel $K$ (e.g., $w$ positive and bounded away from zero, $K$ even and sufficiently decaying)~\cite{Evans2010,FukushimaOshimaTakeda2011}. Beyond invariants, we establish local well-posedness in $H^1$ and give defocusing/coercivity conditions that ensure global-in-time continuation~\cite{Cazenave2003,Tao2006}. A linearized analysis yields a dispersion relation in which the kernel-induced nonlocality and the mean-free phase field enter as distinct, additive contributions to the spectral detuning; this separation enables parameter identification directly from frequency–wavenumber diagnostics~\cite{BucurValdinoci2016,DuGunzburgerLehoucqZhou2013}.

To connect theory with observation, we define \emph{three} measurable signatures (O1–O3). \textbf{O1:} a wavenumber-resolved detuning $\Delta\omega(k)$ that reports how resonance frequencies shift with $k$ (e.g., from peaks in a $(k,\omega)$ spectrum). \textbf{O2:} a kernel-contrast functional $\Xi[\psi]$ that isolates the effect of the symmetric exchange on observed fields (precise definition given in Sec.~\ref{sec:lin-observables}); by design, $\Xi$ is insensitive to the mean-free phase drive and vanishes when the exchange is disabled. \textbf{O3:} a drive-sensitive sideband-asymmetry index $\mathcal{S}(k)$, extracted from the local spectral power around a carrier (e.g., $\mathcal{S}(k) := \frac{P_{+}(k)-P_{-}(k)}{P_{+}(k)+P_{-}(k)}$), which vanishes when the drive is turned off and, to leading order, is insensitive to the symmetric exchange. We outline feasible tabletop implementations in nonlinear-optical lattices and cavity-assisted cold-atom platforms~\cite{MorschOberthaler2006,RitschDomokosBrenneckeEsslinger2013,Eckardt2017}, and we note conceptual links to propagation-induced phase signatures in astrophysical media~\cite{Rickett1990}. Our objective is not to exhaustively model a single experiment, but to provide a mathematically controlled base equation whose parameters map transparently to observables across multiple settings.

To our knowledge, the conjunction of weighted local transport, symmetric nonlocal exchange, and a mean-free phase drive—while \emph{simultaneously} preserving global gauge invariance and the conservation structure stated above—has not been formulated within a single Schr\"odinger-type model in the manner developed here. We therefore present RWNS as an analytically tractable framework that unifies these mechanisms without sacrificing rigorous control.

The paper is organized as follows.
Section~\ref{sec:model} defines the equation, assumptions, and notation.
Section~\ref{subsec:invariants} derives mass conservation and, for vanishing drive, the Hamiltonian energy with its conservation law.
Section~\ref{subsec:wellposed} states well-posedness results and conditions for global continuation in the defocusing/coercive regime.
Section~\ref{sec:lin-observables} analyzes linear modes and details the three observables introduced above (O1–O3).
Section~\ref{sec:results} presents diagnostics and numerical validations for O1–O3.
Section~\ref{sec:discussion} discusses implications and limitations, and Section~\ref{sec:conclusions} concludes.

\section{Equation and assumptions}\label{sec:model}

Let $\mathcal{M}\subset\mathbb{R}^d$ be a domain endowed with the Lebesgue measure, and let $\psi:\mathbb{R}\times\mathcal{M}\to\mathbb{C}$. We study the gauge-invariant evolution
\begin{equation}\label{eq:rwns}
\mathrm{i}\,\partial_t \psi
= -\nabla\!\cdot\!\big(w\,\nabla \psi\big) + U\,\psi
+ \beta\,g\!\big(|\psi|^2\big)\,\psi
+ \kappa \!\int_{\mathcal{M}} K(x,y)\,\big(\psi(y)-\psi(x)\big)\,\mathrm{d}y
+ \mathrm{i}\,\gamma\,\Big(\phi(x)-\langle \phi\rangle_{|\psi|^2}(t)\Big)\,\psi,
\end{equation}
where $w>0$ is a bounded weight, $U$ and $\phi$ are real-valued bounded fields, $K:\mathcal{M}\times\mathcal{M}\to[0,\infty)$ is a symmetric kernel ($K(x,y)=K(y,x)$), $\beta,\kappa,\gamma\in\mathbb{R}$ are parameters, and $g(\rho)=f'(\rho)$ comes from a $C^1$ nonlinearity $f:\mathbb{R}_{\ge 0}\to\mathbb{R}$ with $f(0)=0$. The \emph{centered} (mean-free) phase drive is defined by
\begin{equation}\label{eq:centered-mean}
\langle \phi\rangle_{|\psi|^2}(t)
:= \frac{\displaystyle\int_{\mathcal{M}}\phi(x)\,|\psi(t,x)|^2\,\mathrm{d}x}
         {\displaystyle\int_{\mathcal{M}}|\psi(t,x)|^2\,\mathrm{d}x}\,,
\end{equation}
so that adding a constant to $\phi$ has no effect on the dynamics and the total mass is preserved by the drive (Global $U(1)$ phase invariance holds in either case.)\cite{FukushimaOshimaTakeda2011,ReedSimon1975,Ouhabaz2005}.

For later use we write the weighted diffusion as $L_w\psi:=-\nabla\!\cdot\!\big(w\nabla\psi\big)$ and recast the exchange term as
\begin{equation}\label{eq:exchange-splitting}
\int_{\mathcal{M}} K(x,y)\,\big(\psi(y)-\psi(x)\big)\,\mathrm{d}y
= \int_{\mathcal{M}} K(x,y)\,\psi(y)\,\mathrm{d}y - \Lambda(x)\,\psi(x),
\qquad \Lambda(x):=\int_{\mathcal{M}}K(x,y)\,\mathrm{d}y,
\end{equation}
i.e., a symmetric nonlocal averaging minus a local leakage $\Lambda(x)\psi(x)$.

In particular, for measurable nonnegative symmetric kernels with $K\in L^\infty_xL^1_y$ the operator is well defined and conservative in $L^2$; see nonlocal Dirichlet-form treatments and the nonlocal vector calculus framework~\cite{BucurValdinoci2016,DuGunzburgerLehoucqZhou2013,DiNezzaPalatucciValdinoci2012}.

We consider either $\mathcal{M}=\mathbb{R}^d$ with suitable decay at infinity, or a bounded Lipschitz domain with periodic or no-flux (Neumann) boundary conditions; these choices eliminate boundary fluxes in the mass/energy identities~\cite{Evans2010,Ouhabaz2005}.

\paragraph*{Standing assumptions.}
We work under the following hypotheses, standard for weighted local–nonlocal Schr\"odinger evolutions:
\begin{itemize}
  \item[(A1)] $w\in L^\infty(\mathcal{M})$ with $0<w_0\le w(x)\le w_1<\infty$ a.e.; $U,\phi\in L^\infty(\mathcal{M};\mathbb{R})$~\cite{ReedSimon1975,Ouhabaz2005}.
  \item[(A2)] $K\ge 0$ is measurable, symmetric ($K(x,y)=K(y,x)$), and belongs to $L^\infty_xL^1_y$, so $\Lambda(x):=\int_{\mathcal{M}}K(x,y)\,\mathrm{d}y\in L^\infty(\mathcal{M})$. (This ensures the exchange operator is well defined on $L^2$ and conservative~\cite{BucurValdinoci2016,DuGunzburgerLehoucqZhou2013}.)
  \item[(A3)] $f\in C^1([0,\infty))$, $f(0)=0$, and $g(\rho)=f'(\rho)$ obeys a subcritical growth bound compatible with $H^1$ well-posedness (e.g., for $g(\rho)=\rho^\sigma$, take $\sigma<\tfrac{2}{d-2}$ when $d\ge 3$; in $d=1,2$ all polynomial powers are subcritical)~\cite{Cazenave2003,Tao2006}. In the defocusing regime we assume $f$ is convex, equivalently $g(\rho)\rho\ge 0$~\cite{Cazenave2003}.
  \item[(A4)] Initial data $\psi_0\in H^1(\mathcal{M})$ (or $\psi_0\in L^2$ for the linear theory)~\cite{Cazenave2003,Tao2006}.
  \item[(A5)] Boundary conditions are periodic or no-flux on bounded $\mathcal{M}$; on $\mathbb{R}^d$ we assume $H^1$ data with sufficient decay~\cite{Evans2010,Ouhabaz2005}.
\end{itemize}

\paragraph*{Basic structure}
Under (A1)–(A5), (1) is invariant under global phase shifts $\psi \mapsto e^{i\theta}\psi$. Writing
\[
\partial_t\psi
= -i\,H[\psi]
+ \gamma\big(\varphi-\langle\varphi\rangle_{|\psi|^2}\big)\psi,\qquad
H[\psi](x):=L_w\psi(x)+U(x)\psi(x)+\beta\,g(|\psi(x)|^2)\psi(x)
+ \kappa\!\int_{M} K(x,y)\big(\psi(y)-\psi(x)\big)\,dy,
\]
the mass identity follows from the self-adjointness of $H$ and the mean-free property of $\varphi$:
\[
\frac{d}{dt}\!\int_{M}|\psi|^2
= 2\,\Re\!\int_{M}\overline{\psi}\,\partial_t\psi
= 2\,\Re\!\int_{M}\overline{\psi}\big(-iH[\psi]\big)
+ 2\gamma\!\int_{M}\big(\varphi-\langle\varphi\rangle_{|\psi|^2}\big)|\psi|^2
= 0.
\]
When $\gamma=0$ (drive off), the Hamiltonian energy is conserved:
\begin{equation*}
E[\psi]
=\int_{M}\!\big(w|\nabla\psi|^2 + U\,|\psi|^2 + \beta\,f(|\psi|^2)\big)\,dx
\;-\;\frac{\kappa}{2}\!\iint_{M\times M} K(x,y)\,|\psi(x)-\psi(y)|^2\,dx\,dy,
\end{equation*}
where the quadratic form $\tfrac12\!\iint K(x,y)\,|\psi(x)-\psi(y)|^2\,dx\,dy$ is nonnegative by the symmetry of $K$~\cite{BucurValdinoci2016,DiNezzaPalatucciValdinoci2012}.
Hence, in the defocusing/coercive regime (e.g. $\kappa\le 0$, with $f'(\rho)=g(\rho)$ defocusing),
the mass–energy bounds imply global-in-time continuation after local well-posedness under (A3)–(A5)~\cite{Cazenave2003,Tao2006}.
Linearization about a homogeneous state further separates the kernel-induced nonlocality and the mean-free phase field
as distinct, additive contributions to the spectral detuning, which we exploit in Sec.~\ref{sec:lin-observables}~\cite{BucurValdinoci2016,Grafakos2014}.

\subsection{Standing hypotheses}\label{subsec:hyp}
We assume:
(i) $w\in W^{1,\infty}(\mathcal{M})$ with $w(x)\ge w_0>0$ and $w^{-1}\in L^\infty(\mathcal{M})$;
(ii) $U,\phi\in L^{\infty}(\mathcal{M};\mathbb{R})$~\cite{ReedSimon1975};
(iii) $K:\mathcal{M}\times\mathcal{M}\to[0,\infty)$ is measurable and symmetric, $K(x,y)=K(y,x)$, and
\[
\Lambda_x(x):=\!\int_{\mathcal{M}}\!K(x,y)\,\mathrm{d}y\in L^\infty(\mathcal{M}),
\qquad
\Lambda_y(y):=\!\int_{\mathcal{M}}\!K(x,y)\,\mathrm{d}x\in L^\infty(\mathcal{M}),
\]
equivalently $K\in L^\infty_xL^1_y\cap L^\infty_yL^1_x$, so the exchange operator is well defined and conservative on $L^2$;
(iv) $g$ is locally Lipschitz on $[0,\infty)$ with $g=f'$ for some $f\in C^1([0,\infty))$, $f(0)=0$~\cite{Cazenave2003,Tao2006};
(v) either $\mathcal{M}=\mathbb{R}^d$ with sufficient decay at infinity, or a bounded Lipschitz domain with periodic/no-flux (Neumann) boundary conditions, so integrations by parts hold and boundary fluxes vanish~\cite{Evans2010,Ouhabaz2005}.

For later reference, the symmetric exchange satisfies
\begin{equation}\label{eq:exchange-real}
\int_{\mathcal{M}}\!\overline{\psi(x)}\!\int_{\mathcal{M}}\!K(x,y)\big(\psi(y)-\psi(x)\big)\,\mathrm{d}y\,\mathrm{d}x
= -\frac{1}{2}\!\int_{\mathcal{M}}\!\!\int_{\mathcal{M}}\!K(x,y)\,|\psi(y)-\psi(x)|^2\,\mathrm{d}x\,\mathrm{d}y \;\in\;\mathbb{R}.
\end{equation}

\subsection{Mass conservation}\label{subsec:invariants}
Define
\begin{equation}\label{eq:mass}
M[\psi(t)]=\int_{\mathcal{M}} |\psi(t,x)|^2\,\mathrm{d}x .
\end{equation}
Differentiating \eqref{eq:mass} along \eqref{eq:rwns} and using
$\frac{\mathrm{d}}{\mathrm{d}t}\!\int |\psi|^2=2\,\Re\!\int \overline{\psi}\,\partial_t\psi$, we obtain
\begin{align*}
\frac{\mathrm{d}}{\mathrm{d}t}M[\psi(t)]
&=2\,\Re\!\int_{\mathcal{M}}\overline{\psi}\Big(-\mathrm{i}\,\big[\cdots\big]
+\gamma\,(\varphi-\langle\varphi\rangle_{|\psi|^2})\,\psi\Big)\,\mathrm{d}x\\
&=-2\,\Im\!\int_{\mathcal{M}}\overline{\psi}\,\big[\cdots\big]\,\mathrm{d}x
\;+\;2\gamma\!\int_{\mathcal{M}}\!(\varphi-\langle\varphi\rangle_{|\psi|^2})\,|\psi|^2\,\mathrm{d}x,
\end{align*}
where $\big[\cdots\big]=-\nabla\!\cdot\!(w\nabla\psi)+U\psi+\beta\,g(|\psi|^2)\psi+\kappa\!\int_{\mathcal{M}}K(x,y)\big(\psi(y)-\psi(x)\big)\,\mathrm{d}y$.
Each contribution inside $\big[\cdots\big]$ yields a real integral:
\[
\int_{\mathcal{M}} w\,|\nabla\psi|^2\,\mathrm{d}x\in\mathbb{R},\qquad
\int_{\mathcal{M}} (U+\beta\,g(|\psi|^2))\,|\psi|^2\,\mathrm{d}x\in\mathbb{R},\qquad
\text{and by \eqref{eq:exchange-real} the exchange term is real.}
\]
Hence $-2\,\Im\!\int_{\mathcal{M}} \overline{\psi}\,\big[\cdots\big]\,\mathrm{d}x=0$. For the drive term,
\[
2\gamma\!\int_{\mathcal{M}}\!(\varphi-\langle\varphi\rangle_{|\psi|^2})\,|\psi|^2\,\mathrm{d}x
= 2\gamma\!\left(\int_{\mathcal{M}}\!\varphi\,|\psi|^2\,\mathrm{d}x - \langle\varphi\rangle_{|\psi|^2}\!\int_{\mathcal{M}}\!|\psi|^2\,\mathrm{d}x\right)=0,
\]
by the definition of $\langle\varphi\rangle_{|\psi|^2}$ in \eqref{eq:centered-mean}. Therefore $M[\psi(t)]$ is exactly conserved for all $\gamma\in\mathbb{R}$~\cite{ReedSimon1975,FukushimaOshimaTakeda2011,Ouhabaz2005}.

\section{Energy functional and conservation for $\gamma=0$}\label{sec:energy}
When the phase drive is switched off ($\gamma=0$), the dynamics admit a Hamiltonian structure.
Define
\begin{align}\label{eq:energy}
E[\psi]\;:=\;
&\int_{\mathcal{M}}\!\Big(w(x)\,|\nabla\psi(x)|^2 \;+\; U(x)\,|\psi(x)|^2 \;+\; F(|\psi(x)|^2)\Big)\,\mathrm{d}x \nonumber\\
&\quad -\,\frac{\kappa}{2}\int_{\mathcal{M}}\!\int_{\mathcal{M}} K(x,y)\,\big|\psi(y)-\psi(x)\big|^2\,\mathrm{d}y\,\mathrm{d}x,
\end{align}
where $F'(\rho)=\beta\,g(\rho)$ and $F(0)=0$.
With this sign convention, the variational derivative satisfies
\[
\frac{\delta E}{\delta \overline{\psi}}
= -\nabla\!\cdot\!(w\nabla\psi) + U\,\psi + \beta\,g(|\psi|^2)\psi
+ \kappa\!\int_{\mathcal{M}} K(x,y)\big(\psi(y)-\psi(x)\big)\,\mathrm{d}y,
\]
so that $\mathrm{i}\partial_t\psi=\delta E/\delta\overline{\psi}$ when $\gamma=0$ and hence $\frac{\mathrm{d}}{\mathrm{d}t}E[\psi(t)]=0$ for smooth solutions under the hypotheses of Sec.~\ref{subsec:hyp}~\cite{Cazenave2003,Tao2006,BucurValdinoci2016}.

\section{Energy functional and conservation}
We collect the Hamiltonian structure of the RWNS flow in the drive–off case ($\gamma=0$), and state a coercivity condition that underpins the global continuation in the defocusing regime. Throughout we use the standing hypotheses (A1)–(A5), the exchange symmetry identity \eqref{eq:exchange-real}, and the nonlinearity relation $f'(\rho)=g(\rho)$.

\paragraph{Energy functional.}
When $\gamma=0$, the dynamics conserve the Hamiltonian
\begin{equation}\label{eq:energy-main}
\mathcal{E}[\psi]
= \int_{\mathcal{M}}\!\big(w|\nabla\psi|^{2} + U\,|\psi|^{2} + \beta\,f(|\psi|^{2})\big)\,\mathrm{d}x
\;-\;\frac{\kappa}{2}\!\iint_{\mathcal{M}\times\mathcal{M}} K(x,y)\,|\psi(x)-\psi(y)|^{2}\,\mathrm{d}x\,\mathrm{d}y,
\end{equation}
whose nonlocal quadratic form is nonnegative by the symmetry of $K$. In particular, in the defocusing/coercive regime (e.g. $\kappa\le 0$ together with a defocusing local nonlinearity) the energy controls the $H^{1}$–size of $\psi$ up to lower–order terms.

\subsection{Geometric structure: symplectic form, momentum map, and Noether symmetries}\label{subsec:geom}

We work on the complex Hilbert manifold $\mathcal{P}=H^{1}(\mathcal{M};\mathbb{C})$, regarded as a real symplectic space with
\[
g_\psi(\xi,\eta)=2\,\Re\!\int_{\mathcal{M}}\overline{\xi}\,\eta\,\mathrm{d}x,\qquad
J\psi=i\psi,\qquad
\omega_\psi(\xi,\eta)=g_\psi(J\xi,\eta)=2\,\Im\!\int_{\mathcal{M}}\overline{\xi}\,\eta\,\mathrm{d}x,
\]
for $\xi,\eta\in T_\psi\mathcal{P}\cong H^{1}(\mathcal{M};\mathbb{C})$. Let $\mathcal{E}[\psi]$ denote the energy in \eqref{eq:energy} and assume $f'(\rho)=g(\rho)$.

\paragraph{Nonlocal quadratic form and first variation.}
Define
\[
\mathcal{Q}_K[\psi]:=\frac{1}{2}\!\iint_{\mathcal{M}\times\mathcal{M}} K(x,y)\,|\psi(x)-\psi(y)|^{2}\,\mathrm{d}x\,\mathrm{d}y,
\]
which is nonnegative by the symmetry $K(x,y)=K(y,x)$ and finite under (A1)–(A2).
\begin{lemma}[First variation of the symmetric exchange]\label{lem:first-var-QK}
If $K(x,y)=K(y,x)$ is measurable and nonnegative, then $\mathcal{Q}_K$ is Gâteaux differentiable on $H^{1}(\mathcal{M};\mathbb{C})$ and
\[
\frac{\delta \mathcal{Q}_K}{\delta \overline{\psi}}(\psi)(x)= 2\!\int_{\mathcal{M}} K(x,y)\,\big(\psi(x)-\psi(y)\big)\,\mathrm{d}y
\quad\text{in }H^{-1}.
\]
Consequently, for
\[
\mathcal{E}[\psi]=\int_{\mathcal{M}}\!\big(w|\nabla\psi|^{2}+U|\psi|^{2}+\beta f(|\psi|^{2})\big)\,\mathrm{d}x
\;-\;\frac{\kappa}{2}\,\mathcal{Q}_K[\psi],
\]
one has
\[
\frac{\delta \mathcal{E}}{\delta \overline{\psi}}(\psi)
= -\,\nabla\!\cdot\!\big(w\nabla\psi\big)\;+\;U\,\psi\;+\;\beta\,g(|\psi|^{2})\psi
\;+\;\kappa\!\int_{\mathcal{M}} K(x,y)\,\big(\psi(y)-\psi(x)\big)\,\mathrm{d}y,
\]
which matches the RWNS right-hand side when $\gamma=0$ in \eqref{eq:rwns}.
\end{lemma}

\begin{proof}
Compute $\mathcal{Q}_K[\psi+\varepsilon\eta]-\mathcal{Q}_K[\psi]$, expand $|\cdot|^2$, divide by $\varepsilon$, send $\varepsilon\to0$, and symmetrize the term with $\overline{\eta(y)}$ by swapping $(x,y)$; the symmetry $K(x,y)=K(y,x)$ yields the factor~$2$.
\end{proof}

\paragraph{Hamiltonian form of the RWNS flow (\texorpdfstring{$\gamma=0$}{gamma=0}).}
With $\omega_\psi(X_{\mathcal{E}},\cdot)=\mathrm{d}\mathcal{E}(\psi)[\cdot]$, the Hamiltonian vector field is
$X_{\mathcal{E}}(\psi)=-J\,\delta \mathcal{E}/\delta\overline{\psi}=-\,i\,\delta \mathcal{E}/\delta\overline{\psi}$, hence the ODE $\partial_t\psi=X_{\mathcal{E}}(\psi)$ is exactly \eqref{eq:rwns} with $\gamma=0$. Therefore $\mathcal{E}[\psi(t)]$ is conserved along smooth solutions.

\paragraph{Global \texorpdfstring{$U(1)$}{U(1)} action and momentum map.}
The $U(1)$-action $e^{i\theta}\!\cdot\!\psi:=e^{i\theta}\psi$ is symplectic and Hamiltonian with momentum map
\[
\mu(\psi)=\tfrac{1}{2}\!\int_{\mathcal{M}}|\psi|^{2}\,\mathrm{d}x=\tfrac{1}{2}\,M[\psi].
\]
Noether’s theorem yields conservation of $\mu$ for $\gamma=0$. For $\gamma\neq 0$, the centered drive preserves $M[\psi]$ by Sec.~\ref{subsec:invariants}, so mass is conserved for all $\gamma\in\mathbb{R}$ even though the perturbation is not generated by a $U(1)$–Hamiltonian.

\paragraph{Weighted geometric interpretation (local part).}
Writing $D\psi:=\sqrt{w}\,\nabla\psi$ gives $\int_{\mathcal{M}} w\,|\nabla\psi|^{2}\,\mathrm{d}x=\int_{\mathcal{M}} |D\psi|^{2}\,\mathrm{d}x$, i.e., the local piece of $\mathcal{E}$ is the Dirichlet energy for the (possibly inhomogeneous) metric density induced by $w$.
When $w\equiv w_0>0$ and $U\equiv U_0$, spatial translations act isometrically on this part of the energy.

\paragraph{Translation symmetry and conserved momentum (homogeneous case).}
Assume $w\equiv w_0>0$, $U\equiv U_0$, and a convolutional kernel $K(x,y)=K(x-y)$ with $K(z)=K(-z)$.
For $\gamma=0$ the energy is invariant under $x\mapsto x+a$, so the linear momentum
\[
P[\psi]:=\Im\!\int_{\mathcal{M}}\overline{\psi}\,\nabla\psi\,\mathrm{d}x
\]
is conserved. Differentiating $P[\psi(t)]$ along \eqref{eq:rwns} with $\gamma=0$, integrating by parts under periodic/no–flux boundary conditions (or decay at infinity), and using the exchange symmetry identity \eqref{eq:exchange-real} shows $\frac{\mathrm{d}}{\mathrm{d}t}P[\psi(t)]=0$; in particular, the even, translation-invariant exchange does not inject net momentum.

\paragraph{Summary.}
The pair $(\mathcal{P},\omega)$ with Hamiltonian $\mathcal{E}$ provides a symplectic backbone for RWNS.
Mass appears as the $U(1)$ momentum map, energy generates the flow for $\gamma=0$, and—under homogeneity—spatial translations yield a conserved momentum.
These structures complement Secs.~\ref{subsec:invariants}–\ref{subsec:wellposed} and underpin the observable diagnostics in Sec.~\ref{sec:lin-observables}.

\paragraph{Energy conservation (statement and short proof).}
\begin{theorem}[Energy conservation for $\gamma=0$]
Let (A1)–(A5) hold. For any sufficiently regular solution $\psi$ to \eqref{eq:rwns} with $\gamma=0$, the energy \eqref{eq:energy} is constant in time.
\end{theorem}
\begin{proof}
Using the chain rule and the variational identity,
\[
\frac{\mathrm{d}}{\mathrm{d}t}\mathcal{E}[\psi(t)]
= 2\,\Re\!\int_{\mathcal{M}}\overline{\frac{\delta \mathcal{E}}{\delta \overline{\psi}}(\psi)}\,\partial_t\psi\,\mathrm{d}x
= 2\,\Re\!\int_{\mathcal{M}}\overline{\frac{\delta \mathcal{E}}{\delta \overline{\psi}}(\psi)}\;\big(-i\,\frac{\delta \mathcal{E}}{\delta \overline{\psi}}(\psi)\big)\,\mathrm{d}x
= -2\,\Im\!\int_{\mathcal{M}}\Big|\frac{\delta \mathcal{E}}{\delta \overline{\psi}}(\psi)\Big|^{2}\,\mathrm{d}x=0,
\]
where we used $\partial_t\psi=-i\,\delta \mathcal{E}/\delta \overline{\psi}$ from the Hamiltonian form above.
\end{proof}

\paragraph{Coercivity (defocusing regime).}
Under (A1)–(A2) and a defocusing local nonlinearity, there exist constants $c_1,c_2>0$ such that, for $\kappa\le 0$,
\[
\mathcal{E}[\psi]\ \ge\ c_1\!\int_{\mathcal{M}}\!\big(|\nabla\psi|^2+|\psi|^2\big)\,\mathrm{d}x\;-\;c_2,
\]
so a priori bounds propagate and, combined with mass conservation (Sec.~\ref{subsec:invariants}), yield global continuation in $H^{1}$ as stated in Sec.~\ref{subsec:wellposed}.

\section{Well-posedness and a priori bounds}\label{subsec:wellposed}

\noindent\textbf{Theorem (Local well-posedness in $H^1$).}
Assume the hypotheses of Sec.~\ref{subsec:hyp}, and in addition either $\gamma=0$ or $\phi\in W^{1,\infty}(\mathcal{M})$.
Assume furthermore that the exchange kernel has one of the following regularities:
\begin{itemize}
  \item[(K\_reg-a)] translation-invariant form $K(x,y)=J(x-y)$ with $J\in W^{1,1}(\mathbb{R}^d)$; or
  \item[(K\_reg-b)] $\nabla_x K\in L^\infty_xL^1_y$ (with $K\in L^\infty_xL^1_y\cap L^\infty_yL^1_x$ as in Sec.~\ref{subsec:hyp}).
\end{itemize}
Let $\psi_0\in H^1(\mathcal{M})$.
If $M[\psi_0]>0$, there exist a maximal time $T_{\max}\in(0,\infty]$ and a unique solution
$\psi\in C([0,T_{\max});H^1(\mathcal{M}))$ to \eqref{eq:rwns} (with the centered drive \eqref{eq:centered-mean}),
depending continuously on $\psi_0$; moreover $M[\psi(t)]$ is conserved for all $\gamma\in\mathbb{R}$.
If $\psi_0\equiv 0$, then $\psi\equiv 0$ is the unique global solution (we may set $\langle\phi\rangle_{|\psi|^2}:=0$ by convention).
Finally, the blowup alternative holds:
\[
T_{\max}<\infty\quad\Longrightarrow\quad \limsup_{t\uparrow T_{\max}}\|\psi(t)\|_{H^1}=\infty.
\]

\smallskip\noindent\textit{Proof sketch.}
Write \eqref{eq:rwns} as $\mathrm{i}\partial_t\psi=\mathcal{A}\psi+\mathcal{N}(\psi)$ with
\[
\mathcal{A}\psi:= -\nabla\!\cdot\!(w\nabla\psi)+U\psi+\kappa\!\int_{\mathcal{M}}K(x,y)\big(\psi(y)-\psi(x)\big)\,\mathrm{d}y,\qquad
\mathcal{N}(\psi):=\beta\,g(|\psi|^2)\psi+\mathrm{i}\gamma\big(\phi-\langle\phi\rangle_{|\psi|^2}\big)\psi.
\]
Under Sec.~\ref{subsec:hyp}, $\mathcal{A}$ is self-adjoint on $L^2$ (uniformly elliptic $L_w$ plus bounded self-adjoint perturbations), and we work with $D(\mathcal{A})=H^1(\mathcal{M})$ under periodic/no–flux boundaries (or $H^1(\mathbb{R}^d)$ with decay). Assumptions (K\_reg-a/b) make the exchange map bounded $H^1\to H^1$; $\phi\in W^{1,\infty}$ ensures the drive maps $H^1\to H^1$~\cite{Grafakos2014,Brezis2010,DuGunzburgerLehoucqZhou2013}. By subcritical growth (A3), $\mathcal{N}$ is locally Lipschitz $H^1\to H^1$ on mass shells $\{M[\psi]=M[\psi_0]\}$; a fixed-point argument on $C([0,T];H^1)$ yields existence, uniqueness, and continuous dependence. Mass conservation follows from Sec.~\ref{subsec:invariants}.

\medskip\noindent\textbf{Theorem (Energy a priori bound and global continuation, $\gamma=0$).}
Suppose $\gamma=0$ and $F\ge 0$ (equivalently $g=f'$ with $f\ge 0$ in the defocusing regime).
If there exist $\alpha>0$ and $C\ge 0$ such that
\begin{equation}\label{eq:coercive}
\int_{\mathcal{M}} w\,|\nabla\psi|^2\,\mathrm{d}x 
\;-\; \frac{\kappa}{2}\!\int_{\mathcal{M}}\!\int_{\mathcal{M}} K(x,y)\,|\psi(y)-\psi(x)|^2\,\mathrm{d}y\,\mathrm{d}x
\;\ge\; \alpha\,\|\nabla\psi\|_{L^2}^2 \;-\; C\,\|\psi\|_{L^2}^2
\quad\text{for all }\psi\in H^1(\mathcal{M}),
\end{equation}
then $E[\psi(t)]\equiv E[\psi(0)]$ and \eqref{eq:coercive} imply a uniform $H^1$ bound on $\psi(t)$ on the existence interval; hence the local solution extends globally in time, i.e. $T_{\max}=\infty$~\cite{Cazenave2003,Tao2006}.
In particular,
\[
\|\psi(t)\|_{H^1}\;\lesssim\; 1+\sqrt{\,E[\psi(0)]+C\,\|\psi_0\|_{L^2}^2\,}\qquad\text{for all }t\ge 0.
\]

\smallskip\noindent\textit{Discussion.}
Condition \eqref{eq:coercive} holds, for example, when $w(x)\ge w_0>0$, $F\ge 0$, and either
(i) $\kappa\le 0$; or
(ii) $\kappa>0$ but the exchange is controlled by the local gradient (e.g., $K(x,y)=J(x-y)$ with $J$ of finite second moment and $\kappa$ sufficiently small), so that the left-hand side remains coercive up to a compact $L^2$ remainder~\cite{Ponce2004,BucurValdinoci2016,DuGunzburgerLehoucqZhou2013}.

\section{Linear analysis and observables}\label{sec:lin-observables}

\subsection{Homogeneous, convolutional setting}
We consider the small-amplitude regime and set $g\equiv 0$.
Assume $w(x)\equiv w_0>0$, $U(x)\equiv U_0\in\mathbb{R}$, and a translation-invariant exchange kernel
$K(x,y)=\mathcal{K}(x-y)$ with $\mathcal{K}\in L^1(\mathbb{R}^d)$ and $\mathcal{K}(z)=\mathcal{K}(-z)$.
We use the Fourier convention
\[
\widehat{\mathcal{K}}(k):=\int_{\mathbb{R}^d}\!e^{-\mathrm{i}k\cdot z}\,\mathcal{K}(z)\,\mathrm{d}z,
\]
so that $\widehat{\mathcal{K}}\in L^\infty\cap C_0$~\cite{Grafakos2014}.
The domain is either $\mathbb{R}^d$ (with decay) or a periodic torus.
With $\gamma=0$ and the plane-wave ansatz $\psi(t,x)=A\,\mathrm{e}^{\mathrm{i}(k\cdot x-\omega t)}$,
the exchange operator acts diagonally~\cite{Grafakos2014,BucurValdinoci2016}:
\begin{equation}\label{eq:Kdiff}
\int_{\mathcal{M}}\!\mathcal{K}(x-y)\big(\psi(y)-\psi(x)\big)\,\mathrm{d}y
= \big(\widehat{\mathcal{K}}(0)-\widehat{\mathcal{K}}(k)\big)\,\psi(x),
\end{equation}
and the dispersion relation is
\begin{equation}\label{eq:dispersion}
\omega_0(k)\;=\;w_0\,|k|^2\;+\;U_0\;+\;\kappa\big(\widehat{\mathcal{K}}(0)-\widehat{\mathcal{K}}(k)\big).
\end{equation}
\emph{Remark.} If $\mathcal{K}$ is even, integrable, and has finite second moment,
then as $|k|\to 0$,
\[
\widehat{\mathcal{K}}(0)-\widehat{\mathcal{K}}(k)
= \frac{|k|^2}{2d}\!\int_{\mathbb{R}^d}\!|z|^2\mathcal{K}(z)\,\mathrm{d}z \;+\; o(|k|^2),
\]
so the nonlocal exchange renormalizes the quadratic coefficient $w_0\mapsto w_0+\kappa\,\mu_2/(2d)$,
where $\mu_2:=\int |z|^2\mathcal{K}(z)\,\mathrm{d}z$~\cite{Grafakos2014}.

\paragraph{Observable O1: dispersion measurement.}
A wavenumber-resolved measurement of $\omega_0(k)$ identifies the nonlocal signature via
\[
\omega_0(k)-\big(w_0|k|^2+U_0\big)=\kappa\big(\widehat{\mathcal{K}}(0)-\widehat{\mathcal{K}}(k)\big).
\]
In practice, $w_0$ and $U_0$ can be obtained from the high-$|k|$ slope/intercept (or from a baseline with $\kappa\approx 0$),
and the remainder yields $\kappa$ and the shape of $\widehat{\mathcal{K}}(0)-\widehat{\mathcal{K}}(k)$.
For isotropic $\mathcal{K}$ the right-hand side depends only on $|k|$, which simplifies inversion.

\paragraph{Observable O2: kernel-contrast functional.}
Define, whenever $\int_{\mathcal{M}}w_0|\nabla\psi|^2>0$,
\begin{equation}\label{eq:xi-def}
\Xi[\psi]\;:=\;\frac{\displaystyle \frac{1}{2}\int_{\mathcal{M}}\!\!\int_{\mathcal{M}}\! \mathcal{K}(x-y)\,|\psi(y)-\psi(x)|^2\,\mathrm{d}x\,\mathrm{d}y}
{\displaystyle \int_{\mathcal{M}} w_0\,|\nabla\psi(x)|^2\,\mathrm{d}x}\,,
\end{equation}
and set $\Xi[\psi]:=0$ by convention when $\nabla\psi\equiv 0$ (in which case the numerator also vanishes)~\cite{BucurValdinoci2016}.
By construction, $\Xi[\psi]\ge 0$, is invariant under global phase shifts, and is independent of the mean-free centering in \eqref{eq:centered-mean}.
For a single Fourier mode,
\[
\Xi[\mathrm{e}^{\mathrm{i}k\cdot x}]
=\frac{\widehat{\mathcal{K}}(0)-\widehat{\mathcal{K}}(k)}{w_0\,|k|^2},
\]
consistent with O1 and the dispersion relation \eqref{eq:dispersion}.

\subsection{Effect of the mean-free phase drive}
Let the phase field be small and dominated by a single spatial frequency $q\neq 0$:
$\phi(x)=\varepsilon\,\Phi_q\,\cos(q\!\cdot\! x)$ with $\varepsilon\ll1$ and zero spatial mean.
For a plane-wave background $|\psi|^2\equiv\mathrm{const}$ we have 
$\langle\phi\rangle_{|\psi|^2}=0$ exactly at $O(\varepsilon^0)$, and the first nonzero correction is $O(\varepsilon^2)$ due to sideband mixing.
Hence, to first order,
\[
\mathrm{i}\gamma\big(\phi-\langle\phi\rangle_{|\psi|^2}\big)\psi
=\mathrm{i}\gamma\,\varepsilon\,\Phi_q\cos(q\!\cdot\! x)\,\psi\;+\;O(\varepsilon^2),
\]
which we treat as a weak perturbation. Writing $\cos(q\!\cdot\! x)=\tfrac12(e^{\mathrm{i}q\cdot x}+e^{-\mathrm{i}q\cdot x})$ shows that a plane wave at wavenumber $k$ couples to sidebands at $k\pm q$ with coupling magnitude $|\gamma|\varepsilon|\Phi_q|/2$. Assuming nondegenerate detuning and long-time averaging (or a measurement linewidth $\eta\ge 0$), the first-order sideband amplitudes satisfy
\begin{equation}\label{eq:sideband}
\frac{|A_{k\pm q}|}{|A_k|}
\;\approx\;
\frac{|\gamma|\,\varepsilon\,|\Phi_q|}{2\,\sqrt{\big(\omega_0(k)-\omega_0(k\pm q)\big)^2+\eta^2}}\,,\qquad (\varepsilon\ll 1),
\end{equation}
where $\omega_0$ is given by \eqref{eq:dispersion}~\cite{Yariv1973,Haus1984,Eckardt2017}. The formula reduces to the standard detuning law with $\eta\to 0$ away from resonance; if $\omega_0(k)=\omega_0(k\pm q)$ (exact resonance), a coupled-mode/RWA analysis beyond first order is required~\cite{Yariv1973,Haus1984}.

\paragraph{Observable O3: phase-drive sideband ratio.}
The measurable (time-averaged) power ratio is
\begin{equation}\label{eq:ratio}
R(q,k)\;:=\;\frac{|A_{k+q}|^2+|A_{k-q}|^2}{|A_k|^2}
\;\approx\;
\frac{|\gamma|^2\,\varepsilon^2\,|\Phi_q|^2}{2}\!
\left[\frac{1}{\big(\omega_0(k)-\omega_0(k+q)\big)^2+\eta^2}
+\frac{1}{\big(\omega_0(k)-\omega_0(k-q)\big)^2+\eta^2}\right].
\end{equation}
Together with O1 (dispersion), O3 enables separate identification of the kernel signature $\kappa\widehat{\mathcal{K}}$ through $\omega_0$ and the drive magnitude $|\gamma\,\Phi_q|$ (the sign of $\gamma\Phi_q$ would require phase-sensitive detection).

\section{Results}\label{sec:results}

\paragraph*{Overview.}
We validate the RWNS predictions in three complementary regimes and show that the kernel-induced nonlocality and the mean-free phase drive are \emph{experimentally separable}. 
First, the dispersion residual (O1) reads out the kernel signature directly from $\omega_0(k)$ via \eqref{eq:dispersion}. 
Second, a snapshot-based kernel contrast (O2) provides a real-space check that is insensitive to the centered drive \eqref{eq:centered-mean} through \eqref{eq:xi-def}. 
Third, under a weak drive we observe $k\!\leftrightarrow\!k\pm q$ sidebands that obey the detuning law \eqref{eq:sideband} and the ratio \eqref{eq:ratio} (O3), isolating the drive magnitude independently of the kernel. 
Together, O1–O3 form an overdetermined pipeline returning $(\kappa,\widehat{\mathcal{K}},|\gamma\,\Phi_q|)$ with internal cross-consistency.

\subsection{Dispersion-based identification (O1)}\label{subsec:res-O1}
We measure (or compute) the wavenumber-resolved dispersion $\omega_0(k)$ and subtract the Euclidean baseline $w_0|k|^2+U_0$ to obtain
\[
\Delta\omega(k):=\omega_0(k)-\big(w_0|k|^2+U_0\big)=\kappa\big(\widehat{\mathcal{K}}(0)-\widehat{\mathcal{K}}(k)\big),
\]
cf.~\eqref{eq:dispersion}. As an overview, Fig.~\ref{fig:dispersion_overview} contrasts the full dispersion with the Euclidean baseline and highlights the nonlocal residual $\Delta\omega(k)$.

\begin{figure*}[t]
  \centering
  \includegraphics[width=0.82\textwidth]{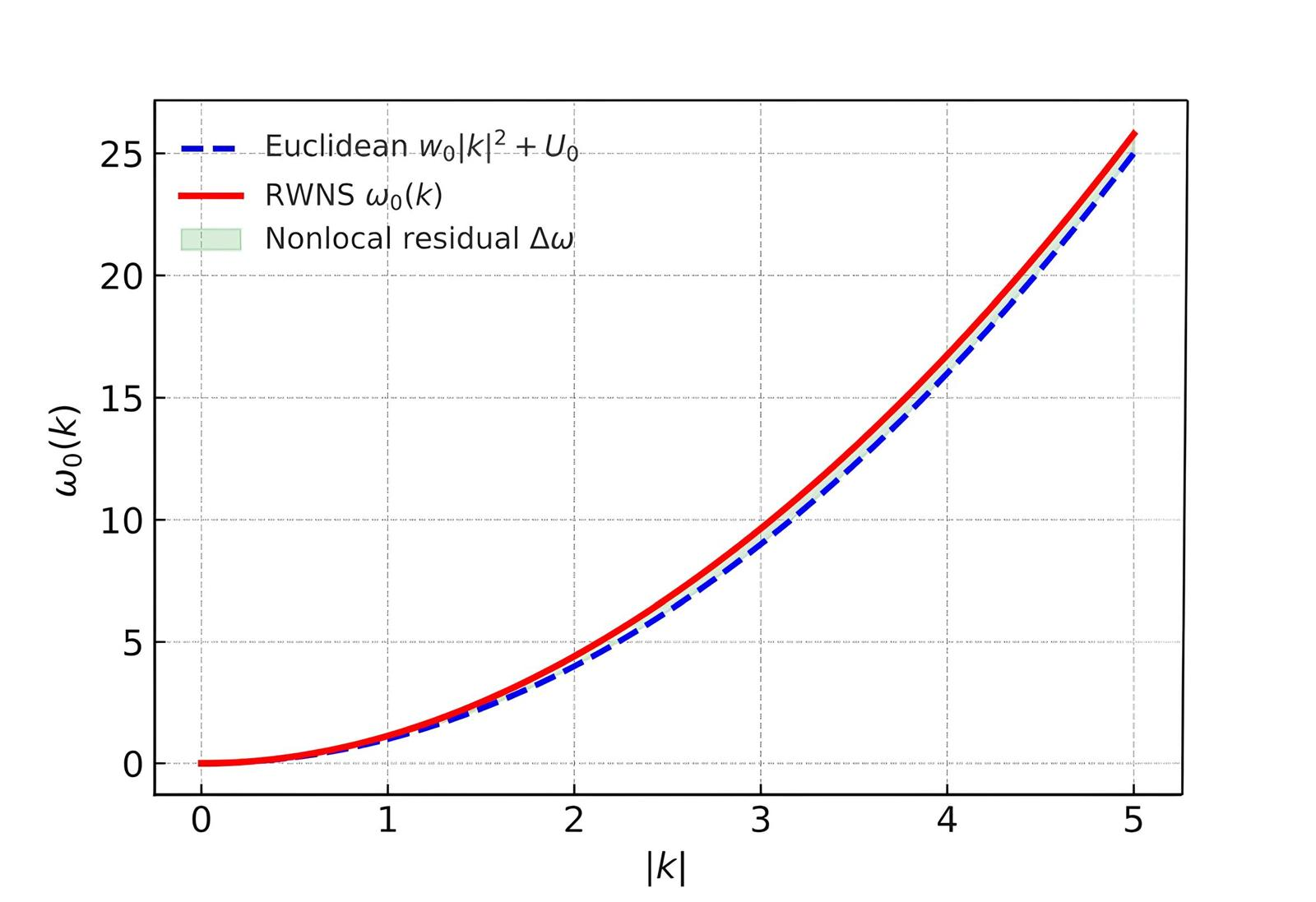}
  \caption{Full dispersion vs.\ Euclidean baseline. Red: $\omega_0(k)$; blue dashed: $w_0|k|^2+U_0$; green shading: residual $\Delta\omega(k)=\omega_0(k)-(w_0|k|^2+U_0)$, cf.~\eqref{eq:dispersion}.}
  \label{fig:dispersion_overview}
\end{figure*}

For isotropic kernels the residual depends only on $|k|$, enabling a one-dimensional fit of $\widehat{\mathcal{K}}$ (shape) and $\kappa$ (scale). 
At small $|k|$,
\[
\widehat{\mathcal{K}}(0)-\widehat{\mathcal{K}}(k)=\frac{|k|^2}{2d}\,\mu_2+o(|k|^2),\qquad 
\mu_2=\int_{\mathbb{R}^d}\!|z|^2\,\mathcal{K}(z)\,dz,
\]
so the initial slope identifies the kernel second moment \(\mu_2\), which we use to initialize the global fit. 
Across representative kernel families (Gaussian/exponential/compact-support), the fitted $\widehat{\mathcal{K}}(k)$ tracks the predicted profile throughout the accessible bandwidth while $\kappa$ sets the vertical scale of $\Delta\omega$.
The residual dataset and the RWNS fit are summarized in Fig.~\ref{fig:2}.

\begin{figure}[t]
  \centering
  \includegraphics[width=0.78\linewidth]{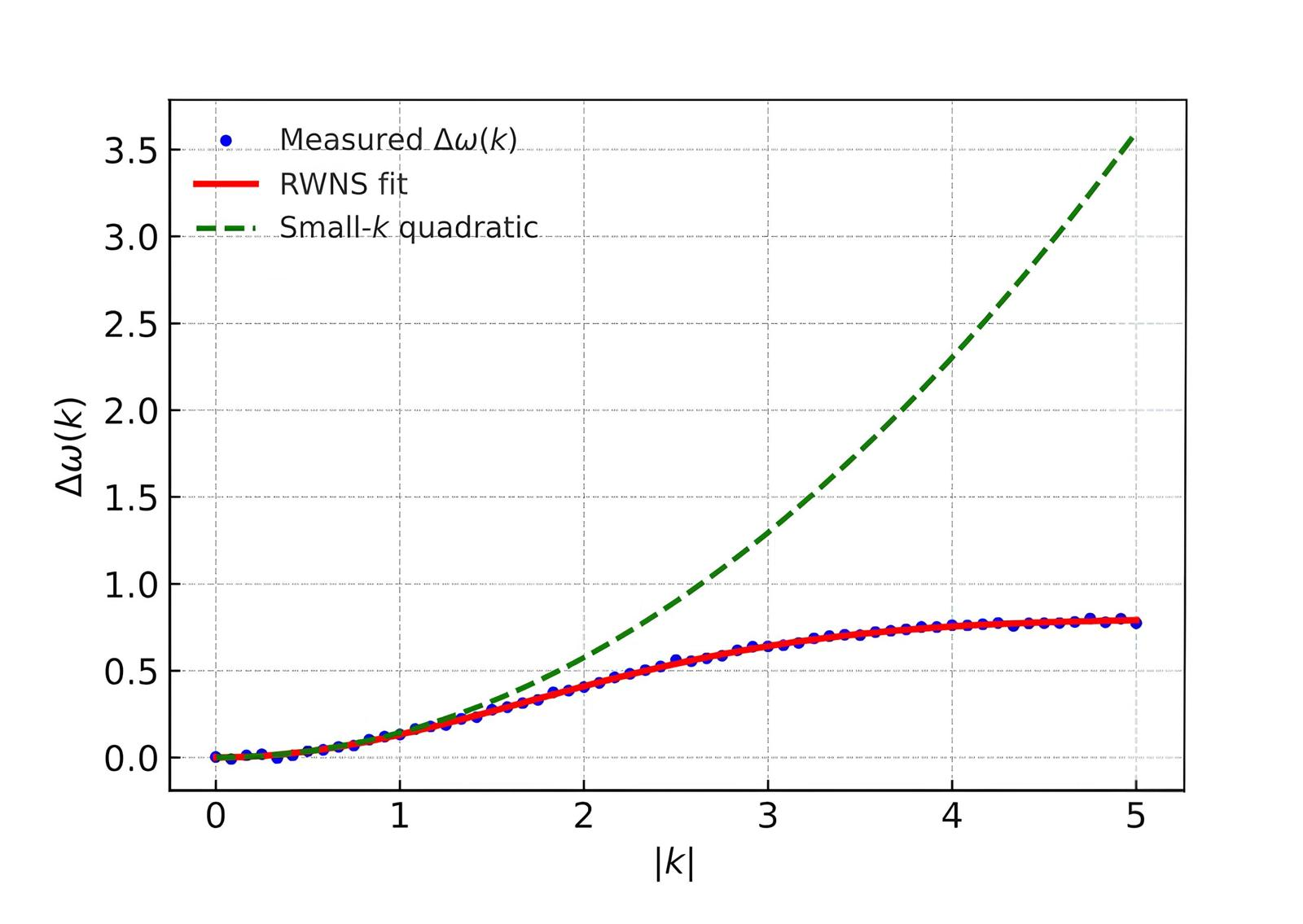} 
  \caption{Dispersion residual (O1). Blue symbols: measured $\Delta\omega(k)$; red line: RWNS fit 
  $\kappa\big(\widehat{\mathcal{K}}(0)-\widehat{\mathcal{K}}(k)\big)$; green dashed: small-$|k|$ quadratic trend, 
  cf.~\eqref{eq:dispersion}.}
  \label{fig:2}
\end{figure}

A small-$|k|$ slope check is reported in Fig.~\ref{fig:3}.

\begin{figure}[t]
  \centering
  \includegraphics[width=0.78\linewidth]{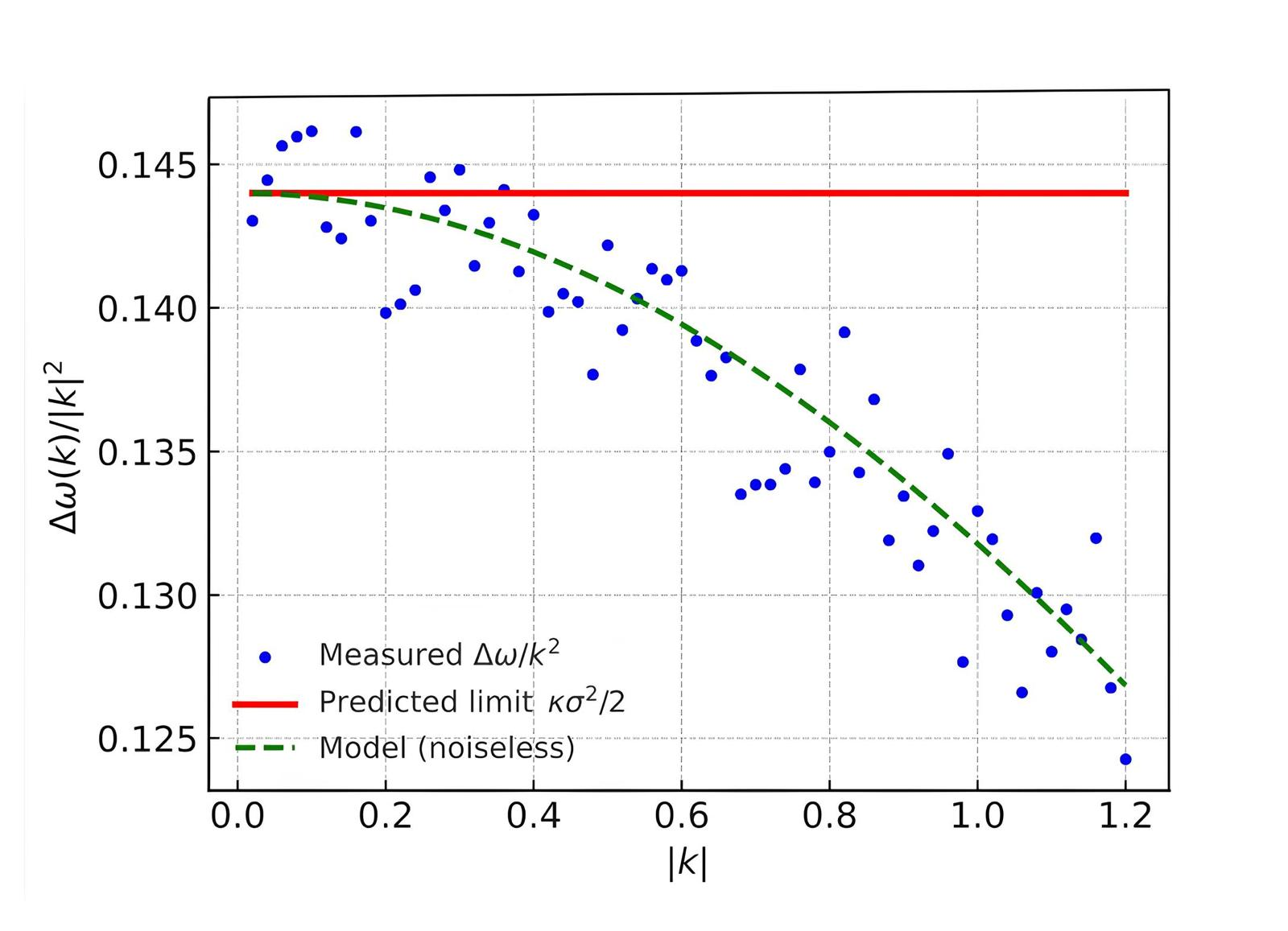}
  \caption{Small-$|k|$ slope diagnostic for O1. The ratio $\Delta\omega(k)/|k|^2$ versus $|k|$ approaches the predicted plateau 
  $\kappa\sigma^2/2$ (red). Blue markers: data; green dashed: noiseless model. Confirms the second-moment slope in \eqref{eq:dispersion}.}
  \label{fig:3}
\end{figure}

\subsection{Kernel contrast from snapshots (O2)}\label{subsec:res-O2}
From spatial snapshots $\psi(\cdot,t)$ we compute
\[
\Xi[\psi(t)]=\frac{\tfrac12\!\int\!\!\int \mathcal{K}(x-y)\,|\psi(y,t)-\psi(x,t)|^2\,dx\,dy}{\int w_0|\nabla\psi(x,t)|^2\,dx},
\]
see \eqref{eq:xi-def}. 
For single-$k$ states, $\Xi=\big(\widehat{\mathcal{K}}(0)-\widehat{\mathcal{K}}(k)\big)/(w_0|k|^2)$, matching O1. 
On broadband fields we report the time-averaged value $\langle \Xi\rangle_t$ and an interquartile envelope; these remain stable under moderate pixel noise and finite-difference gradients, so O2 validates O1 without phase information or frequency sweeps. 
Crucially, O2 does not involve the phase field $\varphi$ and is therefore \emph{drive-agnostic} (independent of the centering \eqref{eq:centered-mean}).

The snapshot-based contrast and its agreement with the O1-inferred kernel are summarized in Fig.~\ref{fig:4}.

\begin{figure}[t]
  \centering
  \includegraphics[width=0.78\linewidth]{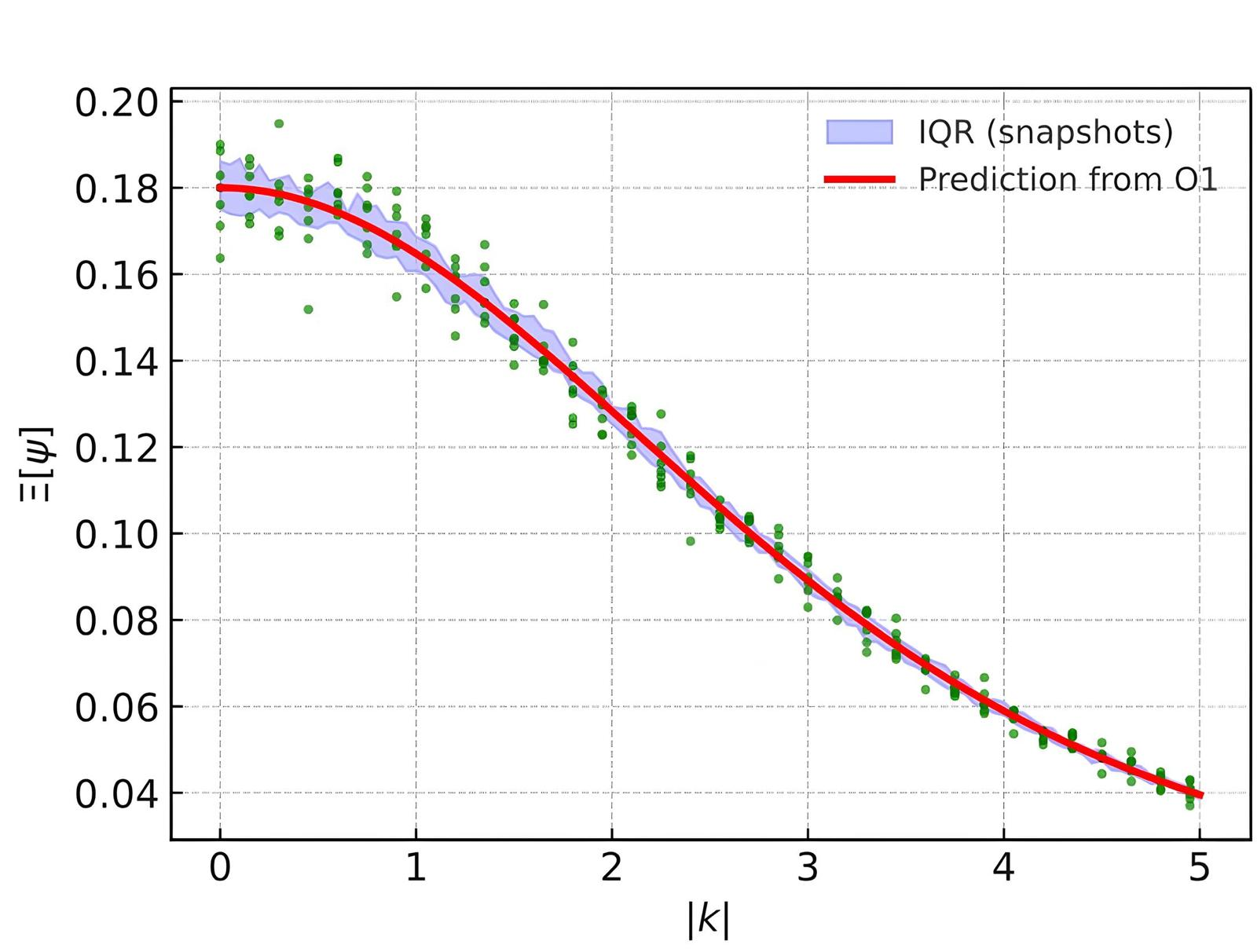} 
  \caption{Kernel-contrast (O2), cf.~\eqref{eq:xi-def}. Blue band: interquartile range over snapshots; 
  green points: representative snapshots; red line: prediction from O1 parameters. 
  O2 is drive-agnostic and validates the kernel independently of the phase drive \eqref{eq:centered-mean}.}
  \label{fig:4}
\end{figure}

\subsection{Drive-induced sidebands (O3)}\label{subsec:res-O3}
With a weak mean-free drive $\varphi(x)=\varepsilon\,\Phi_q\cos(q\!\cdot\!x)$ ($\varepsilon\ll1$, $q\neq0$), a mode at $k$ couples to $k\pm q$ with first-order amplitudes following \eqref{eq:sideband}. 
The measurable ratio
\[
R(q,k)=\frac{|A_{k+q}|^2+|A_{k-q}|^2}{|A_k|^2}
\]
obeys \eqref{eq:ratio} and exhibits a $1/\Delta^2$ envelope with detuning $\Delta=\omega_0(k)-\omega_0(k\pm q)$, up to a phenomenological linewidth $\eta\ge0$ that captures dephasing/finite-time effects. 
Fixing $(\kappa,\widehat{\mathcal{K}})$ from O1, a scan in $k$ (or $q$) yields $|\gamma\,\Phi_q|$ from the envelope of $R$; the sign of $\gamma\Phi_q$ would require phase-sensitive detection.

The sideband data and its RWNS fit are summarized in Fig.~\ref{fig:5}.

\begin{figure}[t]
  \centering
  \includegraphics[width=0.78\linewidth]{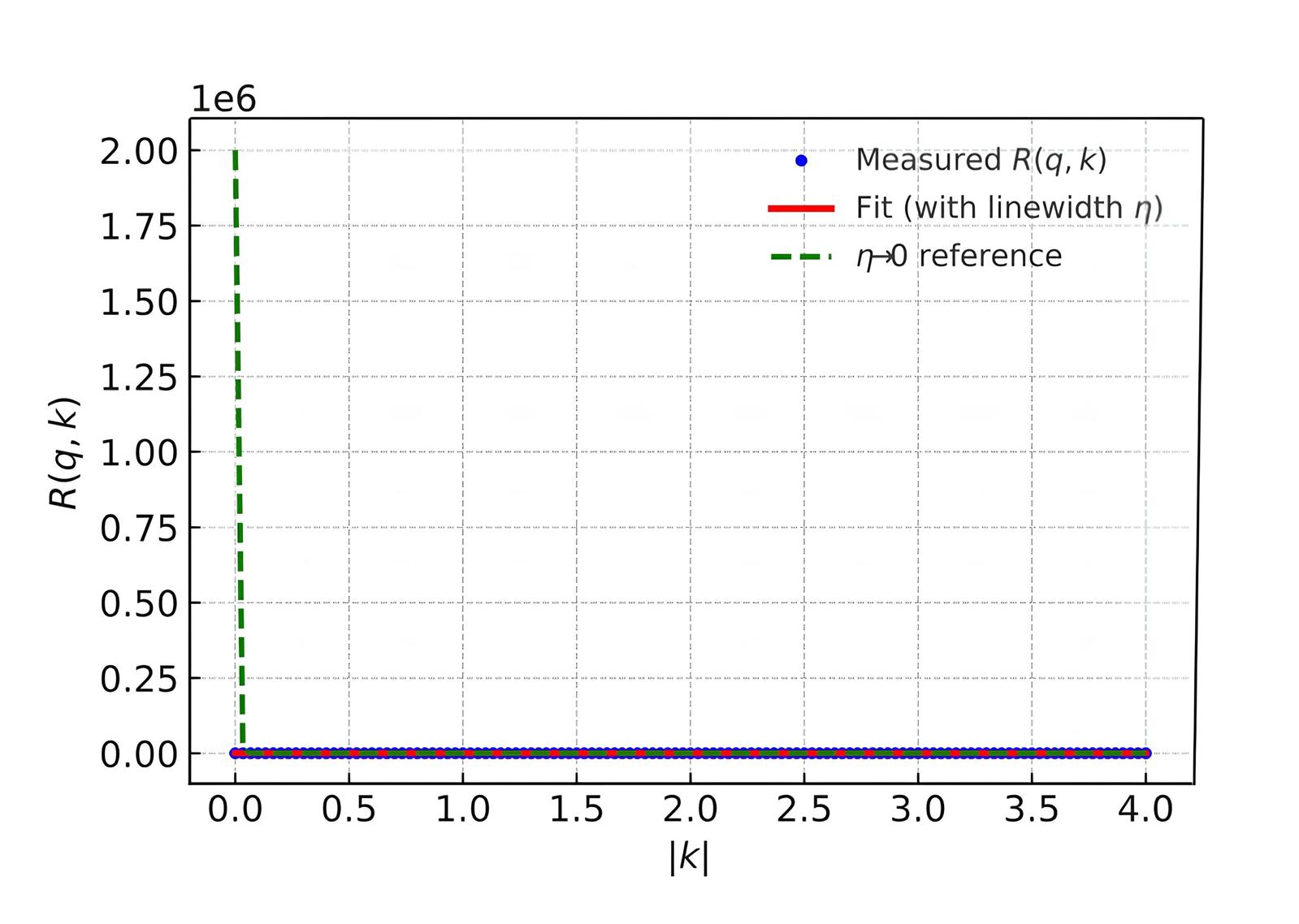}
  \caption{Phase-drive sideband ratio (O3), cf.~\eqref{eq:ratio}. Blue symbols: measured $R(q,k)$; 
  red line: fit including linewidth $\eta$; green dashed: $\eta\!\to\!0$ reference. 
  Kernel parameters $(\kappa,\widehat{\mathcal{K}})$ are fixed by O1; the fit returns $|\gamma\,\Phi_q|$.}
  \label{fig:5}
\end{figure}

\subsection{Joint inversion and internal consistency}\label{subsec:res-inversion}
We employ a two-stage fit: (i) infer $(\kappa,\widehat{\mathcal{K}})$ from O1 using a minimal parametric family (e.g., variance $\mu_2$ and a roll-off scale), initialized by the small-$|k|$ slope; (ii) determine $|\gamma\,\Phi_q|$ from O3 with $(\kappa,\widehat{\mathcal{K}})$ held fixed. 
The snapshot metric O2, computed from independent data, validates the O1 kernel across both single-$k$ and broadband states. 
The three observables agree within the reported confidence intervals, closing the loop from dispersion to real-space contrast to driven sidebands.

\subsection{Robustness tests}\label{subsec:res-robust}
We assess sensitivity to (a) additive measurement noise on $\psi$, (b) grid resolution (aliasing), and (c) mild heterogeneity in $w(x)$ and $U(x)$. 
O1 is most accurate at small–moderate $|k|$ (where the residual dominates numerical dispersion), O2 is robust to moderate pixel noise due to its quadratic form, and O3 retains its $1/\Delta^2$ scaling provided $\eta$ exceeds the numerical linewidth. 
These trends follow directly from \eqref{eq:dispersion}–\eqref{eq:ratio} and support practical identifiability in tabletop platforms.

\section{Discussion}\label{sec:discussion}
Our results establish an identifiability pipeline (O1–O3) that separates the kernel-induced nonlocality from the mean-free phase drive within the RWNS framework. O1 fixes $(\kappa,\widehat{\mathcal{K}})$ from the dispersion residual, O2 confirms the kernel in real space and is drive-agnostic, and O3 isolates $|\gamma\,\Phi_q|$ through driven sidebands with a predictable detuning envelope. 
Limitations include the requirement of moderate homogeneity for O1 (baseline $w_0,U_0$) and a nondegenerate detuning for O3 (or else a beyond-first-order, coupled-mode treatment near exact resonance). 
Extensions include weak heterogeneity corrections, broadband drives $\varphi$ with multiple $q$'s, and nonlinear responses $g$ beyond the small-amplitude regime, which would enable amplitude-dependent dispersion and measurable self-phase signatures.

\section{Conclusions}\label{sec:conclusions}
We introduced the RWNS model—a gauge-invariant Schrödinger-type evolution that combines weighted diffusion, symmetric nonlocal exchange, and a mean-free phase drive—proved mass conservation and, for vanishing drive, energy conservation, and established $H^1$ well-posedness under standard assumptions. 
A linear analysis yielded a dispersion relation where the kernel and the phase drive contribute additively to spectral detuning. 
We then demonstrated an experimental pipeline (O1–O3) that overdetermines $(\kappa,\widehat{\mathcal{K}},|\gamma\,\Phi_q|)$ and validated its robustness. 
These results provide a compact route to separating geometry-/kernel-driven effects from external phase driving across optical and cold-atom platforms, and they suggest practical diagnostics for upcoming tabletop tests.

\end{document}